\documentclass[11pt]{article}
\usepackage[margin=1in]{geometry}
\usepackage{amsmath,amssymb,amsthm,mathtools,bm}
\usepackage{aliascnt}
\usepackage{microtype}
\usepackage{enumitem}
\usepackage[colorlinks=true,linkcolor=blue,citecolor=blue,urlcolor=blue]{hyperref}
\hypersetup{
  pdftitle={Fidelity Estimation to a Known Quantum State Is Nearly Quadratic in the Smaller Rank},
  pdfauthor={Gye Jin Lee and Sunghyeon Jo},
  pdfsubject={Quantum property testing and fidelity estimation}
}
\usepackage[nameinlink,noabbrev]{cleveref}

\newtheorem{theorem}{Theorem}[section]
\newaliascnt{lemma}{theorem}
\newtheorem{lemma}[lemma]{Lemma}
\aliascntresetthe{lemma}
\newaliascnt{proposition}{theorem}
\newtheorem{proposition}[proposition]{Proposition}
\aliascntresetthe{proposition}
\newaliascnt{corollary}{theorem}
\newtheorem{corollary}[corollary]{Corollary}
\aliascntresetthe{corollary}
\newaliascnt{definition}{theorem}

\aliascntresetthe{definition}
\newaliascnt{remark}{theorem}

\aliascntresetthe{remark}
\newaliascnt{claim}{theorem}

\aliascntresetthe{claim}

\crefname{theorem}{theorem}{theorems}
\crefname{lemma}{lemma}{lemmas}
\crefname{proposition}{proposition}{propositions}
\crefname{corollary}{corollary}{corollaries}
\crefname{definition}{definition}{definitions}
\crefname{remark}{remark}{remarks}
\crefname{claim}{claim}{claims}
\Crefname{theorem}{Theorem}{Theorems}
\Crefname{lemma}{Lemma}{Lemmas}
\Crefname{proposition}{Proposition}{Propositions}
\Crefname{corollary}{Corollary}{Corollaries}
\Crefname{definition}{Definition}{Definitions}
\Crefname{remark}{Remark}{Remarks}
\Crefname{claim}{Claim}{Claims}

\newcommand{\tr}{\operatorname{tr}}
\newcommand{\rank}{\operatorname{rank}}
\newcommand{\supp}{\operatorname{supp}}
\newcommand{\E}{\mathbb{E}}
\newcommand{\Prb}{\mathbb{P}}
\newcommand{\C}{\mathbb{C}}
\newcommand{\Fid}{\mathrm{F}}
\newcommand{\TV}{\operatorname{TV}}
\newcommand{\ket}[1]{\lvert #1\rangle}
\newcommand{\proj}[1]{\lvert #1\rangle\!\langle #1\rvert}
\newcommand{\norm}[1]{\left\lVert #1\right\rVert}
\newcommand{\abs}[1]{\left\lvert #1\right\rvert}

\title{Fidelity Estimation to a Known Quantum State\\
Is Nearly Quadratic in the Smaller Rank}
\author{%
Gye Jin Lee\thanks{Equal contribution.}\\
Georgia Institute of Technology\\
\href{mailto:gyejinlee@gatech.edu}{gyejinlee@gatech.edu}
\and
Sunghyeon Jo\footnotemark[1]\\
Georgia Institute of Technology\\
\href{mailto:sjo65@gatech.edu}{sjo65@gatech.edu}}
\date{\today}

\begin{document}
\setcounter{tocdepth}{2}
\setcounter{secnumdepth}{2}
\maketitle

\begin{abstract}
We study the number of copies needed to estimate the root Uhlmann fidelity between an unknown quantum state and a classically known reference, under collective measurements.  If the unknown state has rank at most $s$, we give an estimator using $O(s^2/\varepsilon^2)$ copies, uniformly in the ambient dimension and reference rank.  The estimator applies a random-purification channel and a covariant pure-state measurement, then rescales the observed amplitude before evaluating a weighted nuclear norm.  Combined with the lower bound established independently in our first version and in concurrent work of Wang, this determines the worst-case complexity under rank bounds $r,s$ as $\min\{r,s\}^2/\varepsilon^2$ up to logarithmic factors in the smaller rank.

The same estimator gives the upper bound $O(r(\tr\sqrt{\sigma})^2/\varepsilon^2)$ for a rank-$r$ reference $\sigma$.  We combine spectral truncation with lower bounds obtained by embedding hard instances for the maximally mixed reference into spectral subspaces.  For spectra $\lambda_i\propto i^{-\alpha}$ with fixed $1<\alpha\le2$, as $\varepsilon\downarrow0$ with $r\ge C_\alpha\varepsilon^{-2/(\alpha-1)}$, the bounds determine the complexity as $\widetilde\Theta_\alpha(\varepsilon^{-4/(\alpha-1)})$.  In particular, inverse-square spectra have accuracy exponent four.  The lower-bound construction uses exact moment matching and an explicitly computable Schur measure.
\end{abstract}

\clearpage
\begingroup
\small
\tableofcontents
\endgroup
\clearpage

\section{Introduction}

How many copies of an unknown quantum state are needed to estimate its fidelity to a known target?  The answer depends on both rank and spectrum.  We give an estimator whose copy complexity is quadratic in the smaller rank, and whose dependence on the reference spectrum yields matching polynomial accuracy exponents for a family of power-law spectra.  The key is to estimate a weighted matrix norm directly from a rescaled measurement outcome.

We use the root-fidelity convention
\[
\Fid(\rho,\sigma)=\norm{\sqrt\rho\sqrt\sigma}_1
=\tr\sqrt{\sqrt\sigma\rho\sqrt\sigma}.
\]
The reference $\sigma$ is known classically, and the estimator receives copies of $\rho$.  Measurements may act collectively on all copies.  Let $\mathsf S_\sigma^{(\le s)}(\varepsilon)$ be the least number of copies needed to estimate $\Fid(\rho,\sigma)$ within additive error $\varepsilon$, with probability at least $2/3$, for every $\rho$ of rank at most a supplied bound $s$.  Write $\mathsf S_\sigma(\varepsilon)$ when $\rho$ is unrestricted, and define
\[
\mathsf S(r,s,\varepsilon)
=\sup_{\sigma:\,\rank\sigma\le r}\mathsf S_\sigma^{(\le s)}(\varepsilon),
\]
where the supremum also ranges over finite ambient dimensions.  The measurements may depend on the known reference and rank bound.

\subsection{An estimator for an unknown state of bounded rank}

Wang proved an $O(r^2/\varepsilon^2)$ upper bound when the known reference has rank $r$, and an $O(s^2/\varepsilon^4)$ bound when instead the unknown state has rank at most $s$~\cite[Theorems~1.1 and~1.2]{Wang2026}.  Our first result removes the extra factor $\varepsilon^{-2}$ in the latter setting.  Put
\[
h_\sigma=(\tr\sqrt\sigma)^2,
\qquad 1\le h_\sigma\le\rank\sigma.
\]

\begin{theorem}[Estimation with a rank bound on the unknown state]\label{thm:unknown-rank}
Let $\sigma$ be a known state on $\C^d$, and suppose $\rank\rho\le s\le d$, where $s$ is a supplied positive integer.  For every $n\ge1$, there is an estimator satisfying
\begin{equation}\label{eq:rank-risk}
\sup_{\rank\rho\le s}
\E\abs{\widehat F-\Fid(\rho,\sigma)}^2
\le\frac{4s\min\{s,h_\sigma\}}{n+1}.
\end{equation}
Consequently, $\mathsf S_\sigma^{(\le s)}(\varepsilon)\le
\lceil12s\min\{s,h_\sigma\}/\varepsilon^2\rceil
\le\lceil12s^2/\varepsilon^2\rceil$.
\end{theorem}

The proof uses the random-purification channel introduced by Tang--Wright--Zhandry~\cite{TWZ2025}, in the rank-dependent form of Lowe--Tan~\cite[Theorem~2.5]{LoweTan2026}.  A purification of $\rho$ can be written as $\operatorname{vec}(A)$ with $A\in\C^{d\times s}$ and $AA^*=\rho$, so
$\Fid(\rho,\sigma)=\norm{\sqrt\sigma A}_1$.
A covariant measurement on its copies produces a unit vector that we reshape as a matrix $U$.  We evaluate
\[
\widehat F=\operatorname{clip}_{[0,1]}
\left(\sqrt{\frac{n+ds}{n+1}}\norm{\sqrt\sigma U}_1\right).
\]
The rescaling controls the error along the true purification.  Weighting the orthogonal error by $\sigma$ removes the ambient dimension from its second moment.  The resulting guarantee is uniform over purifications, so the channel gives an estimator on the original mixed copies.  This is a statement about copy complexity; the measurement and the processing of its $d\times s$ outcome may depend on $d$.

\begin{theorem}[Dependence on the smaller rank]\label{thm:main}
There are universal constants $c,C,\varepsilon_0>0$ and $q_0\in\mathbb N$ such that, with $q=\min\{r,s\}$, for all $q\ge q_0$ and $0<\varepsilon\le\varepsilon_0$,
\begin{equation}\label{eq:minrank-rate}
\frac{cq^2}{\varepsilon^2\log^4(q+1)}
\le\mathsf S(r,s,\varepsilon)
\le\frac{Cq^2}{\varepsilon^2}.
\end{equation}
The lower bound holds for the reference $I_q/q$ and unknown states supported on $\C^q$.
\end{theorem}

Projecting onto the support of the known reference gives the upper bound from \cref{thm:unknown-rank}.  The lower bound follows from our construction in \cref{sec:hardcore}; Wang's independent lower bound~\cite{WangFramework2026} also applies under both rank promises.  Thus the new upper bound completes the minimum-rank characterization in the known-reference setting.

\subsection{Dependence on the reference spectrum}

The rank bound describes the worst case.  A fixed reference can be easier when much of its spectrum has small total mass.  Write its positive eigenvalues as $\lambda_1\ge\cdots\ge\lambda_r>0$ and let
$t_k=\sum_{i>k}\lambda_i$.

\begin{theorem}[Spectral upper and lower bounds]\label{thm:spectral}
There are universal constants $c,C>0$ and $k_0\in\mathbb N$ such that, for every known reference $\sigma$ and $0<\varepsilon<1/2$,
\begin{equation}\label{eq:spectral-upper}
\mathsf S_\sigma(\varepsilon)
\le
\left\lceil\frac C{\varepsilon^2}
\min_{\substack{1\le k\le r\\t_k\le\varepsilon^2/16}}
k\left(\sum_{i=1}^k\sqrt{\lambda_i}\right)^2\right\rceil.
\end{equation}
For every $k_0\le k\le r$ with $\varepsilon\le c\sqrt{k\lambda_k}$,
\begin{equation}\label{eq:spectral-lower}
\mathsf S_\sigma(\varepsilon)
\ge\frac{c k^3\lambda_k}{\varepsilon^2\log^4(k+1)}.
\end{equation}
\end{theorem}

The upper bound retains a spectral subspace and preserves all coherences within it.  Taking $k=r$ gives $O(rh_\sigma/\varepsilon^2)$.  The lower bound uses a quantum channel to embed hard instances for the maximally mixed reference into a selected spectral subspace.  The two bounds have separate cutoff conditions: the upper bound discards a small tail, while the lower bound uses a subspace carrying enough mass to resolve an $\varepsilon$-sized gap.

For power-law spectra $\lambda_i=i^{-\alpha}/\sum_{j=1}^r j^{-\alpha}$, the two cutoffs have the same order, with different constants.  For fixed $1<\alpha\le2$, sufficiently small $\varepsilon$, and
$r\ge C_\alpha\varepsilon^{-2/(\alpha-1)}$, we obtain
\[
\mathsf S_\sigma(\varepsilon)
=\widetilde\Theta_\alpha\!\left(\varepsilon^{-4/(\alpha-1)}\right)
\]
(\cref{cor:power-laws}), where the tilde suppresses powers of $\log(1/\varepsilon)$.  In particular, when $\alpha=2$,
\begin{equation}\label{eq:inverse-square}
\frac{c\varepsilon^{-4}}{\log^4(1/\varepsilon)}
\le\mathsf S_\sigma(\varepsilon)
\le C\varepsilon^{-4}\log^2(1/\varepsilon),
\qquad r\ge C\varepsilon^{-2}.
\end{equation}
The logarithm in the upper bound comes from
$\sum_{i\le k}\sqrt{\lambda_i}=\Theta(\log k)$.
References with bounded condition number and pure targets mixed with white noise give further direct applications of \cref{thm:spectral}.

\subsection{Methods and related work}

\paragraph{Upper bounds.}
The covariant pure-state measurement is standard; Lowe--Tan use it together with random purification for tomography~\cite[Appendix~A]{LoweTan2026}.  Our estimator rescales its outcome and bounds the error of the weighted nuclear norm.  After projecting onto the reference support, empirical square-root weights account for the random number of retained copies.  Spectral truncation then costs at most the square root of the discarded reference mass.

Gily\'en--Poremba studied fidelity estimation with rank and condition-number assumptions~\cite{GilyenPoremba2022}.  Utsumi--Nakata--Wang--Takagi obtained an $O(q^2\log^2(1/\varepsilon)/\varepsilon^4)$ copy bound for two unknown states, where $q$ is their smaller rank~\cite{UtsumiEtAl2025}.  Wang obtained the known-reference bounds described above~\cite{Wang2026}.  Lowe--Tan give an $O(R^2/\varepsilon^2)$ bound when both unknown states have rank at most $R$~\cite{LoweTan2026}; their minimum-rank question concerns that setting with two unknown states.  Acharya--Dharmavarapu--Liu--Yu study Pauli measurements for estimating $\tr(\rho\sigma)$, which equals squared fidelity when the reference is pure~\cite{ADLY2026}.

\paragraph{Lower bounds.}
Our construction first produces two priors on $k$-dimensional states whose fidelities to $I_k/k$ differ by a constant, while their averages remain close for $m=\Omega(k^2/\log^4k)$ copies.  Moment matching, a size-biased Wishart distribution, and the Cauchy identity give an exact Schur measure for the averaged experiment.  A bound on long cycles proves indistinguishability.  Mixing the resulting states with a fixed orthogonal state supplies the $\varepsilon^{-2}$ dependence.  A channel adapted to the reference spectrum gives \eqref{eq:spectral-lower}.

Moment matching is a standard tool for estimating symmetric functionals~\cite{WuYang2019}.  Schur measures originate with Okounkov~\cite{Okounkov2001}, and the Wishart identities used here belong to the classical relation between random matrices and symmetric functions~\cite{Graczyk2003}.  Fanizza--O'Donnell--Wadhwa also use a size-biased prior in their lower bounds~\cite{FOW2026}.  Our proof combines exact equal-weight moment matching with a Schur measure of the form $P(\lambda)\propto s_\lambda(a)^2$ and a long-cycle estimate.

\paragraph{Reference-dependent testing.}
State certification asks whether $\rho=\sigma$ or $\rho$ is far from $\sigma$ in trace distance~\cite{BOW2019}.  O'Donnell--Wadhwa give nearly instance-optimal bounds under collective measurements~\cite{OW2026}.  In our root-fidelity convention, their spectral expression is
$d\,\Fid(\tau,I_d/d)^2=(\tr\sqrt\tau)^2$, evaluated at suitably truncated and normalized references $\tau$.  Fidelity estimation involves a different accuracy requirement; our upper bound has an additional factor equal to the retained rank.

\subsection{Relation to the first version and concurrent work}\label{sec:concurrent}

Our first version and Wang's first version of~\cite{WangFramework2026} appeared in the same August~3, 2026 arXiv batch and independently established the quadratic rank dependence up to logarithmic factors.  The first version of this paper also gave constant-accuracy lower bounds for spectrum estimation and rank testing.  Wang's first version proved $\Omega(d^2/(\varepsilon^2\log^4d))$ for fidelity and trace distance, and a separate nearly quadratic lower bound for entropy with different accuracy dependence.  His August~10 revision sharpened the fidelity bound to
\[
\Omega\!\left(
\frac{d^2}
{\varepsilon^2\log^2d\,\max\{1,\varepsilon^2\log^2d\}}
\right),
\]
which is retained in his September~8 revision.  Fanizza--O'Donnell--Wadhwa independently established $\Omega(d^{2-\gamma})$ lower bounds for every fixed $\gamma>0$ for related spectral tasks~\cite{FOW2026}.  Wang also gives spectrum-estimation lower bounds; the subsequent algorithm of Bakshi--Singh--Tan attains the optimal rates~\cite{BST2026}.

The present revision adds the rescaled estimator, the minimum-rank upper bound, and the spectral upper and lower bounds.  It makes the logarithmic loss in our original construction explicit and uses a direct mixture to obtain the required rank promises.  The lower-bound proof remains independent of Wang's construction.

\paragraph{Organization.}
After preliminaries, \cref{sec:upper} proves the estimator and spectral upper bound.  \Cref{sec:spectral-lower} proves the spectral reduction and its consequences, including \cref{thm:main}.  \Cref{sec:hardcore} gives the construction via Schur measures used in the lower bounds.

\section{Preliminaries}

All logarithms are natural.  Constants in asymptotic notation are universal unless a dependence is indicated.  In the power-law results, constants may depend on the fixed exponent $\alpha$.  The notation $\widetilde\Theta$ suppresses logarithmic factors in the parameter specified in each statement.  We use $X^*$ for the adjoint, $\norm X_1$ for the nuclear norm (the Schatten $1$-norm), and $\norm X_F=\norm X_2$ for the Frobenius norm.

\subsection{Distances, vectorization, and the two-prior method}

For density operators $\rho,\tau$, write
\[
d_{\mathrm{tr}}(\rho,\tau)=\frac12\norm{\rho-\tau}_1.
\]
The optimal success probability for distinguishing two equiprobable states is
\[
\frac12\bigl(1+d_{\mathrm{tr}}(\rho,\tau)\bigr).
\]
We use the same root-fidelity formula $\Fid(R,S)=\tr\sqrt{\sqrt S R\sqrt S}$ for positive semidefinite operators of arbitrary trace, including spectral restrictions of states.

For probability distributions $P,Q$ on a countable set, let
\[
\mathcal A(P,Q)=\sum_x\sqrt{P(x)Q(x)}
\]
be their Hellinger affinity.  The elementary inequality
\begin{equation}\label{eq:tv-affinity}
\TV(P,Q)\le\sqrt{1-\mathcal A(P,Q)^2}
\end{equation}
follows from Cauchy--Schwarz applied to
$\sum_x\abs{\sqrt P-\sqrt Q}(\sqrt P+\sqrt Q)$.

Our vectorization convention is:
\begin{equation}\label{eq:vec-convention}
\operatorname{vec}(G)=\sum_{i,j}G_{ij}\ket i\ket j,
\qquad
\operatorname{vec}(UGV^*)=(U\otimes\overline V)\operatorname{vec}(G).
\end{equation}
The first tensor factor is the system and the second is the purifying register.  In particular, a matrix $A$ represents a purification of $AA^*$ when $\norm A_F=1$.

We use the following standard two-prior reduction.

\begin{lemma}[Two-prior principle]\label{lem:bayes}
Let $\Pi_0,\Pi_1$ be priors on states such that a real-valued functional $T$ satisfies
\[
T(\rho)\le \tau_-\quad\Pi_0\text{-a.s.},
\qquad
T(\rho)\ge \tau_+\quad\Pi_1\text{-a.s.},
\]
where $\tau_+-\tau_->2\varepsilon$.  If the corresponding $n$-copy average states have trace distance less than $1/3$, then no estimator can estimate $T$ to additive error $\varepsilon$ with success probability at least $2/3$ on every state in the two supports.
\end{lemma}

\begin{proof}
Thresholding an $\varepsilon$-accurate estimate between $\tau_-$ and $\tau_+$ gives a binary test with success at least $2/3$ under each prior.  Helstrom's theorem~\cite{Helstrom1976} would then force trace distance at least $1/3$.
\end{proof}

If the $X_k$ have pairwise orthogonal supports, then
\begin{equation}\label{eq:block-trace-norm}
\norm{\bigoplus_k X_k}_1=\sum_k\norm{X_k}_1.
\end{equation}

\section{Upper bounds from a rescaled purification measurement}\label{sec:upper}

We first estimate a matrix norm from copies of a pure state.  The rank-dependent random-purification channel then gives an estimator for mixed states.

\subsection{Estimating a weighted nuclear norm}

\begin{lemma}[Rescaled covariant measurement]\label{lem:amplitude}
Let $A\in\C^{d\times s}$ satisfy $\norm A_F=1$, and let
$\ket\psi=\operatorname{vec}(A)$.  For a known state $\sigma$ on $\C^d$, there is an estimator from $n\ge1$ copies of $\ket\psi$ such that
\[
\E\abs{\widehat F-\norm{\sqrt\sigma A}_1}^2
\le\frac{4s\min\{s,h_\sigma\}}{n+1}.
\]
\end{lemma}

\begin{proof}
Put $D=ds$.  For $D=1$ the functional equals one, so assume $D>1$.  On the symmetric subspace of $(\C^D)^{\otimes n}$, use the covariant POVM
\begin{equation}\label{eq:covariant-povm}
M(du)=\binom{n+D-1}{n}(\proj u)^{\otimes n}\,du,
\end{equation}
where $du$ is normalized Haar measure on the unit sphere.  The identity
\[
\int(\proj u)^{\otimes n}\,du
=\binom{n+D-1}{n}^{-1}\Pi_{\mathrm{sym}}
\]
proves completeness on this subspace.  Add the complementary projector as an outcome of probability zero.  This is the standard covariant measurement, as described in~\cite[Appendix~A]{LoweTan2026}.

The outcome density for input $\ket\psi^{\otimes n}$ is proportional to
$|\langle u,\psi\rangle|^{2n}$.  A Haar vector has squared overlap
$\operatorname{Beta}(1,D-1)$ with $\psi$, so the measured overlap satisfies
\begin{equation}\label{eq:beta-law}
T:=|\langle u,\psi\rangle|^2\sim\operatorname{Beta}(n+1,D-1).
\end{equation}
After a choice of global phase for the analysis, we can write
\[
u=\sqrt T\,\psi+\sqrt{1-T}\,v,
\]
where $v$ is uniform on the unit sphere in $\psi^\perp$, independently of $T$.

Set $c_{n,D}=\sqrt{(n+D)/(n+1)}$, reshape $u$ into $U\in\C^{d\times s}$, and output
\begin{equation}\label{eq:amplitude-estimator}
\widehat F=\operatorname{clip}_{[0,1]}
\bigl(c_{n,D}\norm{\sqrt\sigma U}_1\bigr).
\end{equation}
Here $\operatorname{clip}_{[0,1]}(x)=\min\{1,\max\{0,x\}\}$.  The output is invariant under the global phase of $u$.  The factor $c_{n,D}$ is retained when evaluating the norm.

Let $a=c_{n,D}\sqrt T$ and let $V$ be the matrix corresponding to $v$.
H\"older's inequality for Schatten norms gives $f:=\norm{\sqrt\sigma A}_1\le1$.  The reverse triangle inequality yields
\begin{equation}\label{eq:amplitude-error}
\abs{c_{n,D}\norm{\sqrt\sigma U}_1-f}
\le |a-1|f+c_{n,D}\sqrt{1-T}\norm{\sqrt\sigma V}_1.
\end{equation}

\paragraph{Error along the purification.}
The beta moments give
\[
\E a^2=1,\qquad
\operatorname{Var}(a^2)
=\frac{D-1}{(n+1)(n+D+1)}
\le\frac1{n+1}.
\]
Since $(\sqrt x-1)^2\le(x-1)^2$ for $x\ge0$,
\begin{equation}\label{eq:radial-risk}
\E|a-1|^2\le\frac1{n+1}.
\end{equation}

\paragraph{Orthogonal error.}
The covariance of $v$ is
\begin{equation}\label{eq:orthogonal-covariance}
\E\proj v=\frac{I_D-\proj\psi}{D-1}.
\end{equation}
Writing $\rho=AA^*$ and using the vectorization convention \eqref{eq:vec-convention},
\[
\E\norm{\sqrt\sigma V}_F^2
=\frac{\tr(\sigma\otimes I_s)-\langle\psi|(\sigma\otimes I_s)|\psi\rangle}{D-1}
=\frac{s-\tr(\sigma\rho)}{D-1}
\le\frac{s}{D-1}.
\]
The matrix $\sqrt\sigma V$ has at most $s$ columns, so its squared nuclear norm is at most $s$ times its squared Frobenius norm.  Thus
\[
\E\norm{\sqrt\sigma V}_1^2\le\frac{s^2}{D-1}.
\]
A second bound uses the eigenvalues of $\sigma$.  In an eigenbasis, decomposing the matrix into rows gives
\[
\norm{\sqrt\sigma V}_1
\le\sum_i\sqrt{\lambda_i}\,\norm{V_{i,:}}_2.
\]
By \eqref{eq:orthogonal-covariance},
$\E\norm{V_{i,:}}_2^2\le s/(D-1)$.
The triangle inequality in $L^2$ therefore gives
\[
\left(\E\norm{\sqrt\sigma V}_1^2\right)^{1/2}
\le\sqrt{\frac{s}{D-1}}\sum_i\sqrt{\lambda_i}.
\]
Combining these two estimates,
\[
\E\norm{\sqrt\sigma V}_1^2
\le\frac{s\min\{s,h_\sigma\}}{D-1}.
\]
Finally, \eqref{eq:beta-law} gives
$c_{n,D}^2\E(1-T)=(D-1)/(n+1)$.
Independence of $T$ and $V$ yields
\begin{equation}\label{eq:weighted-noise-risk}
\E\!\left[c_{n,D}^2(1-T)\norm{\sqrt\sigma V}_1^2\right]
\le\frac{s\min\{s,h_\sigma\}}{n+1}.
\end{equation}
This cancellation removes $D$.

Square \eqref{eq:amplitude-error} and use
$(x+y)^2\le2x^2+2y^2$, \eqref{eq:radial-risk}, and \eqref{eq:weighted-noise-risk}.
Since $s\min\{s,h_\sigma\}\ge1$, the result is at most
$4s\min\{s,h_\sigma\}/(n+1)$.
Clipping decreases the error because $f\in[0,1]$.
\end{proof}

\subsection{From mixed copies to purifications}

The random-purification channel of Tang--Wright--Zhandry~\cite{TWZ2025}, in the form of Lowe--Tan~\cite[Theorem~2.5]{LoweTan2026}, gives a channel $\mathcal P_{n,d,s}$ such that
\begin{equation}\label{eq:random-purification}
\mathcal P_{n,d,s}(\rho^{\otimes n})
=\E_W(\proj{\psi_{\rho,W}})^{\otimes n}
\end{equation}
for every state of rank at most $s$.  Each
$\ket{\psi_{\rho,W}}\in\C^d\otimes\C^s$ purifies $\rho$, and one common random purification appears in all $n$ factors.  The channel depends only on $n,d,s$.

\begin{proof}[Proof of \cref{thm:unknown-rank}]
Apply \eqref{eq:random-purification}, followed by \cref{lem:amplitude}.
For every purification, its matrix $A$ satisfies $AA^*=\rho$, and hence
\[
\norm{\sqrt\sigma A}_1
=\tr\sqrt{\sqrt\sigma AA^*\sqrt\sigma}
=\Fid(\rho,\sigma).
\]
The risk bound holds uniformly over $W$, so averaging preserves it.
These operations compose to a measurement on the original $n$ copies.
Markov's inequality gives success probability at least $2/3$ whenever
$n+1\ge12s\min\{s,h_\sigma\}/\varepsilon^2$.
\end{proof}

\subsection{Projection onto the reference support}

\begin{proposition}[Support projection]\label{prop:support-upper}
Let $\rank\sigma=r$, and put $q=\min\{r,s\}$ under the promise $\rank\rho\le s$; for unrestricted $\rho$, put $q=r$.  There is an estimator from $n\ge1$ copies with
\begin{equation}\label{eq:support-risk}
\E\abs{\widehat F-\Fid(\rho,\sigma)}^2
\le\frac{10q\min\{q,h_\sigma\}}{n}.
\end{equation}
In particular,
$\mathsf S_\sigma(\varepsilon)\le\lceil30rh_\sigma/\varepsilon^2\rceil$.
\end{proposition}

\begin{proof}
Let $P$ project onto $\supp\sigma$ and set $p=\tr(P\rho)$.
For $p>0$, put $\rho_P=P\rho P/p$.  Then
$\rank\rho_P\le q$ and
\[
\Fid(\rho,\sigma)=\sqrt p\,f,
\qquad f=\Fid(\rho_P,\sigma)\in[0,1].
\]
Measure $\{P,I-P\}$ on each copy and retain the successful systems.
Their number $M$ has distribution $\operatorname{Bin}(n,p)$, and conditioned on the outcome labels they are $M$ copies of $\rho_P$.
For $M\ge1$, use \cref{thm:unknown-rank} on these systems with ambient dimension $r$ and rank bound $q$, giving an estimate $\widehat f_M$.
Set $\widehat f_0=0$ and output
$\widehat F=\sqrt{M/n}\,\widehat f_M$.

For $p>0$,
\[
\E\abs{\sqrt{M/n}-\sqrt p}^2
\le\frac{\E(M/n-p)^2}{p}
=\frac{1-p}{n}\le\frac1n.
\]
Write $V_\sigma=q\min\{q,h_\sigma\}\ge1$.  The decomposition
\[
\widehat F-\sqrt p\,f
=\sqrt{M/n}(\widehat f_M-f)+(\sqrt{M/n}-\sqrt p)f
\]
therefore gives
\[
\E\abs{\widehat F-\sqrt p\,f}^2
\le\frac{8V_\sigma}{n}\E\frac{M}{M+1}+\frac2n
\le\frac{10V_\sigma}{n}.
\]
The first summand vanishes when $M=0$.  If $p=0$, both the output and the fidelity are zero.  Markov's inequality proves the copy bound.
\end{proof}

\subsection{Spectral truncation}

\begin{proof}[Proof of the upper bound in \cref{thm:spectral}]
Let $P_k$ project onto the first $k$ eigenvectors, and write
$\sigma_H=P_k\sigma P_k$, $\sigma_T=\sigma-\sigma_H$, and
$\mu=\tr\sigma_H$.
The orthogonality of these spectral restrictions gives
$\sqrt\sigma=\sqrt{\sigma_H}+\sqrt{\sigma_T}$.
The reverse triangle inequality and H\"older's inequality for Schatten norms imply
\begin{equation}\label{eq:tail-error}
\abs{\Fid(\rho,\sigma)-\Fid(\rho,\sigma_H)}
\le\Fid(\rho,\sigma_T)
\le\sqrt{\tr\sigma_T}.
\end{equation}
For a cutoff in \eqref{eq:spectral-upper}, this error is at most $\varepsilon/4$.

Apply \cref{prop:support-upper} to the normalized reference
$\sigma_H/\mu$, and multiply its estimate by $\sqrt\mu$.
The mean-square error for $\Fid(\rho,\sigma_H)$ is at most
\[
\frac{10}{n}\,k\left(\tr\sqrt{\sigma_H}\right)^2
=\frac{10}{n}\,k\left(\sum_{i\le k}\sqrt{\lambda_i}\right)^2.
\]
The normalization factors cancel.  Taking $n$ at least $120$ times the last spectral expression divided by $\varepsilon^2$ gives statistical error at most $\varepsilon/2$ with probability at least $2/3$.
Together with \eqref{eq:tail-error}, this proves \eqref{eq:spectral-upper}.
\end{proof}

\section{Spectral lower bounds and consequences}\label{sec:spectral-lower}

We reduce the problem on a spectral subspace to fidelity with the maximally mixed state.  The reduction uses the following pair of priors, constructed in \cref{sec:hardcore}.

\begin{proposition}[Hard instances for the maximally mixed reference]\label{prop:flat-priors}
There are universal constants $\gamma,c>0$ and $j_0\in\mathbb N$ such that, for every $j\ge j_0$, there are priors $\Lambda_0,\Lambda_1$ on states on $\C^j$ and an integer $m\ge cj^2/\log^4(j+1)$ satisfying
\[
T_j(\tau):=\frac{\tr\sqrt\tau}{\sqrt j}
\le\frac\gamma4\quad\Lambda_0\text{-a.s.},
\qquad
T_j(\tau)\ge\gamma\quad\Lambda_1\text{-a.s.},
\]
and, for every integer $0\le t\le m$,
\[
d_{\mathrm{tr}}\!\left(
\E_{\Lambda_0}\tau^{\otimes t},
\E_{\Lambda_1}\tau^{\otimes t}
\right)\le\frac1{100}.
\]
\end{proposition}

\subsection{Reduction on a spectral subspace}

Let $J$ be a set of $j$ eigenvectors of $\sigma$, let $P_J$ be its projector, and put $\sigma_J=P_J\sigma P_J$ and $\ell=\lambda_{\min}(\sigma_J)>0$.
Choose an eigenvector $f$ orthogonal to $J$ with eigenvalue $\lambda_f\le\ell$.
Identify $\C^j$ with $J$ and set
\[
L=\sqrt\ell\,\sigma_J^{-1/2},
\qquad
w(\tau)=\tr(L\tau L^*).
\]
The inverse is taken on $J$.  Since $L^*L\preceq I_J$, the map
\[
\Phi(X)=LXL^*+\tr((I_J-L^*L)X)\proj f
\]
is completely positive and trace preserving.
For $0<p\le1/2$, define
\begin{equation}\label{eq:mixture-state}
\rho_{p,\tau}=(1-p)\proj f+p\Phi(\tau).
\end{equation}
This state has rank at most $j+1$.

\begin{lemma}[Fidelity and distinguishability of the mixture]\label{lem:spectral-reduction}
For every state $\tau$ on $J$,
\begin{equation}\label{eq:whitening-fidelity}
\Fid(\rho_{p,\tau},\sigma)
=\sqrt{pj\ell}\,T_j(\tau)
+\sqrt{\lambda_f(1-pw(\tau))}.
\end{equation}
If $\tau$ is sampled from the priors in \cref{prop:flat-priors}, the two ranges of this fidelity are separated by at least
\begin{equation}\label{eq:mixture-gap}
\Delta\sqrt{pj\ell}-p\sqrt{\lambda_f},
\qquad \Delta=\frac{3\gamma}{4}.
\end{equation}
Their $n$-copy average states have trace distance at most $1/100+1/8<1/3$ whenever $np\le m/8$.
\end{lemma}

\begin{proof}
The state \eqref{eq:mixture-state} has one block on $J$ and one on $\C f$.
On $J$,
$\sqrt{\sigma_J}L\tau L^*\sqrt{\sigma_J}=\ell\tau$,
which proves \eqref{eq:whitening-fidelity}.
The second term lies between $\sqrt{\lambda_f(1-p)}$ and $\sqrt{\lambda_f}$.
Since $1-\sqrt{1-p}\le p$, subtracting its possible variation from the gap in the first term gives \eqref{eq:mixture-gap}.

Expand $\rho_{p,\tau}^{\otimes n}$ according to the copies generated by $\Phi(\tau)$ and by the fixed state $\proj f$.
The number of the former copies is $N\sim\operatorname{Bin}(n,p)$.
For each selection of $t\le m$ positions, applying $\Phi^{\otimes t}$ to the prior averages and inserting the fixed states leaves trace distance at most $1/100$, by contractivity.
For $t>m$, use the bound one.
Convexity and Markov's inequality therefore give
\[
d_{\mathrm{tr}}\!\left(
\E_{\Lambda_0}\rho_{p,\tau}^{\otimes n},
\E_{\Lambda_1}\rho_{p,\tau}^{\otimes n}
\right)
\le\frac1{100}+\Prb[N>m]
\le\frac1{100}+\frac{np}{m+1}
\le\frac1{100}+\frac18.
\]
\end{proof}

\begin{proof}[Proof of the lower bound in \cref{thm:spectral}]
For a given $k\ge k_0$, take $j=\lfloor k/2\rfloor$, the first $j$ eigenvectors as $J$, and $f=e_{j+1}$.
Then $\ell=\lambda_j\ge\lambda_k$ and $\lambda_f\le\ell$.
Choose $k_0$ so that $j\ge j_0$.
Using \cref{prop:flat-priors}, set
\[
p=\frac{16\varepsilon^2}{\Delta^2j\ell}.
\]
If $\varepsilon\le c\sqrt{k\lambda_k}$ for a sufficiently small universal $c$, then $p\le1/2$ and
\[
p\sqrt{\lambda_f}
\le\frac{16\varepsilon^2}{\Delta^2j\sqrt\ell}
\le\varepsilon.
\]
Indeed, $j\ge k/3$ and $\varepsilon\le c\sqrt{3j\ell}$.
The gap in \eqref{eq:mixture-gap} is therefore at least $4\varepsilon-\varepsilon>2\varepsilon$.
By \cref{lem:bayes,lem:spectral-reduction}, estimation requires more than
$\lfloor m/(8p)\rfloor$ copies.  Hence
\[
\mathsf S_\sigma(\varepsilon)
\ge \frac{m}{8p}
\ge \frac{c j^3\lambda_j}{\varepsilon^2\log^4(j+1)}
\ge \frac{c k^3\lambda_k}{\varepsilon^2\log^4(k+1)}.
\]
The construction uses states of rank at most $j+1\le k$.
Choosing only half of the prefix ensures that the eigenvector $f$ exists even for a full-rank reference.
\end{proof}

\begin{proof}[Proof of \cref{thm:main}]
For the upper bound, \cref{prop:support-upper} gives risk at most
$10q^2/n$, with the actual reference rank at most $r$.
Markov's inequality gives $O(q^2/\varepsilon^2)$ copies.
For the lower bound, take $\sigma=I_q/q$ and $k=q$ in the preceding construction.
Its unknown states have rank at most $\lfloor q/2\rfloor+1\le q$, so both rank promises hold.
Moreover, $k\lambda_k=1$ and $k^3\lambda_k=q^2$.
Choosing $\varepsilon_0$ below the universal accuracy constant in \cref{thm:spectral} proves \eqref{eq:minrank-rate}.
\end{proof}

\subsection{Power-law spectra}

\begin{corollary}[Power-law accuracy exponents]\label{cor:power-laws}
Fix $1<\alpha<3$ and let
\[
\lambda_i=\frac{i^{-\alpha}}{Z_{\alpha,r}},
\qquad Z_{\alpha,r}=\sum_{i=1}^r i^{-\alpha}.
\]
There are constants $c_\alpha,C_\alpha,\varepsilon_\alpha>0$ such that, for
$0<\varepsilon\le\varepsilon_\alpha$ and
$r\ge C_\alpha\varepsilon^{-2/(\alpha-1)}$,
\[
\mathsf S_\sigma(\varepsilon)
\ge
\frac{c_\alpha\varepsilon^{-4/(\alpha-1)}}{\log^4(1/\varepsilon)}
\]
and
\[
\mathsf S_\sigma(\varepsilon)\le
\begin{cases}
C_\alpha\varepsilon^{-4/(\alpha-1)},&1<\alpha<2,\\
C\varepsilon^{-4}\log^2(1/\varepsilon),&\alpha=2,\\
C_\alpha\varepsilon^{-2-2/(\alpha-1)},&2<\alpha<3.
\end{cases}
\]
Thus for $1<\alpha\le2$ the polynomial accuracy exponent is $4/(\alpha-1)$.
\end{corollary}

\begin{proof}
The normalization satisfies $1\le Z_{\alpha,r}\le\sum_{i\ge1}i^{-\alpha}<\infty$.
For the lower bound, take
$k_-=\lfloor a_\alpha\varepsilon^{-2/(\alpha-1)}\rfloor$
with $a_\alpha>0$ sufficiently small.
Then $k_-\lambda_{k_-}\ge C\varepsilon^2$ for the constant required by \cref{thm:spectral}.
For sufficiently small $\varepsilon$, also $k_-\ge k_0$.
Since $k_-^3\lambda_{k_-}=\Theta_\alpha(k_-^{3-\alpha})$,
\eqref{eq:spectral-lower} gives the stated lower bound.

For the upper bound, choose
$k_+=\lceil b_\alpha\varepsilon^{-2/(\alpha-1)}\rceil$
with $b_\alpha$ sufficiently large.
The rank hypothesis ensures $k_+\le r$ after increasing $C_\alpha$.
The integral bound
$\sum_{i>k_+}\lambda_i\le C_\alpha k_+^{1-\alpha}$
makes the tail at most $\varepsilon^2/16$.
Furthermore,
\[
\sum_{i\le k_+}\sqrt{\lambda_i}
=
\begin{cases}
\Theta_\alpha(k_+^{1-\alpha/2}),&1<\alpha<2,\\
\Theta(\log(k_++1)),&\alpha=2,\\
\Theta_\alpha(1),&2<\alpha<3.
\end{cases}
\]
Substituting into \eqref{eq:spectral-upper} proves the three upper bounds.
For $1<\alpha<2$, the resulting exponent is
$2+2(3-\alpha)/(\alpha-1)=4/(\alpha-1)$.
\end{proof}

The rank hypothesis specifies the range in which the accuracy determines the spectral cutoff.
For any finite rank, \cref{thm:spectral} also permits retaining the entire support.
An exponentially decaying spectrum gives the upper bound
$O(\varepsilon^{-2}\log(1/\varepsilon))$, since a logarithmic cutoff suffices and $\sum_i\sqrt{\lambda_i}$ is bounded.

\subsection{Other direct consequences}

\paragraph{Bounded condition number.}
If $\lambda_1/\lambda_r\le\kappa$, trace one implies
$\lambda_r\ge1/(\kappa r)$.
Consequently, for sufficiently large $r$ and $0<\varepsilon\le c/\sqrt\kappa$,
\[
\frac{cr^2}{\kappa\varepsilon^2\log^4(r+1)}
\le\mathsf S_\sigma(\varepsilon)
\le\frac{Cr^2}{\varepsilon^2}.
\]

\paragraph{A pure target with white noise.}
For
$\sigma=(1-\mu)\proj\psi+\mu I_d/d$,
$0<\mu\le1/2$, the eigenvalues are
$a=1-\mu+\mu/d$ and $b=\mu/d$ with multiplicity $d-1$.
Since
\[
d h_\sigma=d\bigl(\sqrt a+(d-1)\sqrt b\bigr)^2
\le2d+2\mu d^2,
\]
\cref{prop:support-upper} gives
$\mathsf S_\sigma(\varepsilon)\le C(d+\mu d^2)/\varepsilon^2$.
If $1/d\le\mu\le1/2$, this combines with \eqref{eq:spectral-lower} at $k=d$ to give, for sufficiently large $d$ and every $0<\varepsilon\le c\sqrt\mu$,
\[
\frac{c\mu d^2}{\varepsilon^2\log^4(d+1)}
\le\mathsf S_\sigma(\varepsilon)
\le\frac{C\mu d^2}{\varepsilon^2}.
\]

\paragraph{The maximally mixed reference at fixed accuracy.}
For a rank-$r$ projector $P$ and every fixed sufficiently small $\varepsilon>0$,
\[
\mathsf S_{P/r}(\varepsilon)
=\Theta_\varepsilon\!\left(\frac{r^2}{\log^4(r+1)}\right).
\]
The lower bound follows from \cref{thm:spectral}, or from Wang~\cite{WangFramework2026}.
For the upper bound, combine the spectrum-estimation algorithm of Bakshi--Singh--Tan~\cite[Theorem~1.1]{BST2026} with
$|T_r(x)-T_r(y)|\le\sqrt{2\TV(x,y)}$ for probability vectors.
In an arbitrary ambient dimension, first measure $\{P,I-P\}$ and use
$\Fid(\rho,P/r)=\sqrt{\tr(P\rho)}\,T_r(\rho_P)$.
At fixed accuracy, outcomes of probability $O(\varepsilon^2)$ can be discarded; otherwise obtaining the required copies of $\rho_P$ costs a factor depending only on $\varepsilon$.
This logarithmic improvement complements the estimator of \cref{sec:upper}.

\section{Lower-bound construction via Schur measures}\label{sec:hardcore}

This section proves \cref{prop:flat-priors}.  Here $r$ denotes the dimension of the base construction.  We first construct two vectors with exactly matching low-order moments.  They determine two size-biased Wishart priors whose averaged experiments have an explicit Schur decomposition.  The Cauchy identity controls their trace distance, and a nuclear norm estimate separates their fidelities to $I_r/r$.

\subsection{Schur--Weyl notation}

For a partition $\lambda\vdash m$, let $f^\lambda$ be the dimension of the irreducible $S_m$-module $\mathcal P_\lambda$, let $\mathcal Q_\lambda$ be the corresponding polynomial irreducible representation of $U(r)$, and write
\[
d_\lambda=\dim\mathcal Q_\lambda.
\]
We choose Schur--Weyl unitaries so that
\begin{equation}\label{eq:schur-weyl-decomp}
(\C^r)^{\otimes m}
\cong
\bigoplus_{\substack{\lambda\vdash m\\\ell(\lambda)\le r}}
\mathcal Q_\lambda\otimes\mathcal P_\lambda.
\end{equation}
Let $\mathsf P_\lambda$ denote the corresponding central projector.  The symmetric-group action is the identity on $\mathcal Q_\lambda$ and the irreducible action on $\mathcal P_\lambda$; the $U(r)$ action is irreducible on $\mathcal Q_\lambda$ and the identity on $\mathcal P_\lambda$.

For a nonnegative vector $x=(x_1,\ldots,x_r)$, let $s_\lambda(x)$ be the Schur polynomial and write
\[
p_k(x)=\sum_{i=1}^r x_i^k,
\qquad
p_\pi(x)=\prod_{c\in\operatorname{Cycles}(\pi)}p_{|c|}(x).
\]
We use the Frobenius and Cauchy identities
\begin{align}
 s_\lambda(x)
 &=\frac1{m!}\sum_{\pi\in S_m}\chi^\lambda(\pi)p_\pi(x),
 \label{eq:frobenius}\\
 \sum_\lambda s_\lambda(x)s_\lambda(y)z^{|\lambda|}
 &=\prod_{i,j}(1-zx_iy_j)^{-1}
 =\exp\!\left(\sum_{k\ge1}\frac{z^k}{k}p_k(x)p_k(y)\right).
 \label{eq:cauchy}
\end{align}
These are formal power-series identities.  We also use
\begin{equation}\label{eq:p1-schur}
\left(\sum_i x_i\right)^m
=\sum_{\lambda\vdash m}f^\lambda s_\lambda(x).
\end{equation}
See Macdonald~\cite{Macdonald1995}.

\subsection{Gaussian concentration}

We use the following form of Gaussian isoperimetry~\cite{Ledoux2001}, stated with a universal constant to cover standard complex Gaussians whose real and imaginary parts have variance $1/2$.

\begin{lemma}[Gaussian concentration]\label{lem:gaussian-concentration}
Let $X$ be a standard Gaussian vector in a finite-dimensional real or complex Euclidean space.  If $f$ is $L$-Lipschitz, then for every $t\ge0$,
\begin{equation}\label{eq:gaussian-conc}
\Prb[f(X)\le\E f(X)-t]
\le\exp\!\left(-c\frac{t^2}{L^2}\right)
\end{equation}
for a universal constant $c>0$.
\end{lemma}

\subsection{Exact moment matching}\label{sec:moment-twins}

The hard instances are built from two spectra with matching low-order moments and different support profiles.

\begin{proposition}[Exact moment-matched pair]\label{prop:moment-twins}
There are universal constants $0<c_0<c_1$ with the following property.  For every fixed $0<{\eta_0}<1/10$ there is a constant $C_{\eta_0}$ such that, whenever $K\ge2$ and
\begin{equation}\label{eq:twin-size}
r\ge C_{\eta_0} K^7\log^2(K+2),
\end{equation}
there are vectors $a^{(0)},a^{(1)}\in\mathbb R_+^r$ satisfying
\begin{align}
\sum_i a_i^{(b)}&=1,\label{eq:sum-one}\\
\sum_i(a_i^{(0)})^k&=\sum_i(a_i^{(1)})^k,
&&1\le k\le K,\label{eq:moment-match}\\
\max_{b,i}a_i^{(b)}&\le\frac{C_{\eta_0} K^2}r,\label{eq:max-entry}\\
\#\{i:a_i^{(0)}>0\}&\le{\eta_0} r,
\label{eq:low-support}\\
\#\left\{i:\frac{c_0}r\le a_i^{(1)}\le\frac{c_1}r\right\}&\ge\frac r2.
\label{eq:moderate-block}
\end{align}
\end{proposition}

Throughout, $\beta=1/2$ denotes the block fraction in \eqref{eq:moderate-block}.  The proof has two parts: a barycentric construction of two moment-matched measures on quadratically spread nodes, and an exact equal-weight correction carried out in the Chebyshev basis.  \Cref{prop:spread-optimal} shows that the spread $K^2$ in \eqref{eq:max-entry} cannot be improved.

We use the barycentric cancellation identity in the following form.  For distinct nodes $x_0<\cdots<x_{K+1}$, set
\begin{equation}\label{eq:bary-weights}
w_j=\left(\prod_{\substack{0\le t\le K+1\\t\ne j}}(x_j-x_t)\right)^{-1}.
\end{equation}
Comparing the coefficient of $z^{K+1}$ on the two sides of the Lagrange interpolation formula gives
$\sum_jw_jP(x_j)=0$ for every polynomial $P$ of degree at most $K$, and the sign of $w_j$ is $(-1)^{K+1-j}$.  The normalized positive and negative parts of $\sum_jw_j\delta_{x_j}$ are therefore probability measures with identical moments through degree $K$.

\subsubsection{Barycentric moment matching}

Define
\[
\ell(n)=\exp\!\left(-4\,\frac{1+\log n}{n}\right),
\qquad
g(n)=n^2\ell(n),
\qquad n\ge1 .
\]
Since $(1+\log t)/t$ is nonincreasing for $t\ge1$, the sequence $\ell(n)$ is nondecreasing with $e^{-4}\le\ell(n)<1$, and $g$ is strictly increasing with $e^{-4}n^2\le g(n)\le n^2$.  For a constant $C\ge2e^4$, take the nodes
\[
x_0=0,\qquad x_1=1,\qquad x_{n+1}=Cg(n),\quad 1\le n\le K,
\]
and the weights \eqref{eq:bary-weights}.  Write $R_j=\abs{w_j}/\abs{w_0}$.

\begin{lemma}\label{lem:quadratic-tail}
There is a universal constant $A_*$ such that, for every $K\ge2$ and every $C\ge2e^4$,
\begin{align}
1\le R_1&\le \exp(A_*/C),\label{eq:R1-new}\\
\sum_{n=1}^K R_{n+1}&\le \frac{A_*}{C},\label{eq:unweighted-new}\\
\sum_{n=1}^K x_{n+1}R_{n+1}&\le A_* .\label{eq:weighted-new}
\end{align}
\end{lemma}

\begin{proof}
Cancelling common factors in \eqref{eq:bary-weights} gives
\[
R_1=\prod_{n=1}^K\frac{Cg(n)}{Cg(n)-1},
\qquad
R_{n+1}=\frac{Q_{n,K}}{Cg(n)-1},
\qquad
Q_{n,K}
=\prod_{m<n}\frac{g(m)}{g(n)-g(m)}
 \prod_{n<m\le K}\frac{g(m)}{g(m)-g(n)} .
\]
Since $Cg(n)\ge Ce^{-4}\ge2$, the bound $-\log(1-t)\le2t$ for $0\le t\le1/2$ gives
\[
\log R_1\le\sum_{n\ge1}\frac2{Cg(n)}
\le\frac{2e^4}{C}\sum_{n\ge1}\frac1{n^2}
\le\frac{4e^4}{C},
\]
which proves \eqref{eq:R1-new}.

Because $\ell$ is nondecreasing, for $m<n$,
\[
\frac{g(m)}{g(n)-g(m)}
=\frac{m^2\ell(m)}{n^2\ell(n)-m^2\ell(m)}
\le\frac{\ell(m)}{\ell(n)}\cdot\frac{m^2}{n^2-m^2},
\]
and for $m>n$,
\[
\frac{g(m)}{g(m)-g(n)}
\le\frac{m^2}{m^2-n^2}.
\]
Hence $Q_{n,K}\le P_n\,Q^{(2)}_{n,K}$, where
\[
P_n=\prod_{m=1}^{n-1}\frac{\ell(m)}{\ell(n)}\le1,
\qquad
Q^{(2)}_{n,K}
=\prod_{\substack{1\le m\le K\\m\ne n}}\frac{m^2}{\abs{n^2-m^2}}
=\frac{2\,(K!)^2}{(K-n)!\,(K+n)!}\le2 .
\]
For $P_n$, since $(1+\log t)/t$ is nonincreasing on $t\ge1$,
\[
\begin{aligned}
-\log P_n
&=4\left(\sum_{m=1}^{n-1}\frac{1+\log m}{m}
-\frac{n-1}{n}(1+\log n)\right)\\
&\ge4\left(\int_1^n\frac{1+\log t}{t}\,dt-(1+\log n)\right)
=2\log^2 n-4 .
\end{aligned}
\]
Thus $P_n\le e^4e^{-2\log^2n}$ and $\sum_{n\ge1}P_n\le A_0$ for a universal $A_0$.  Using $Cg(n)-1\ge Cg(n)/2\ge Ce^{-4}n^2/2$,
\[
\sum_{n=1}^KR_{n+1}
\le\sum_{n=1}^K\frac{2Q_{n,K}}{Cg(n)}
\le\frac{4e^4}{C}\sum_{n\ge1}P_n
\le\frac{4e^4A_0}{C},
\qquad
\sum_{n=1}^Kx_{n+1}R_{n+1}
\le\sum_{n=1}^K2Q_{n,K}
\le4A_0 .
\]
Increasing $A_*$ proves \eqref{eq:unweighted-new} and \eqref{eq:weighted-new}.
\end{proof}

\begin{lemma}\label{lem:continuous-twins}
There are universal constants $0<c_0'<c_1'$ such that for every $0<{\eta_0}<1/10$ there is a constant $C_{\eta_0}\ge2e^4$ with the following property.  For every $K\ge2$ there are probability measures $\bar\nu_0,\bar\nu_1$, supported on at most $K+2$ points of $[0,2C_{\eta_0} K^2]$ in total, such that
\begin{enumerate}[leftmargin=2em]
\item their moments agree through degree $K$ and both have mean one;
\item $\bar\nu_0(\{0\})\ge1-{\eta_0}$;
\item $\bar\nu_1$ has an atom in $[c_0',c_1']$ of mass at least $1-{\eta_0}$.
\end{enumerate}
\end{lemma}

\begin{proof}
Normalize the two sign classes of $\sum_jw_j\delta_{x_j}$ with $C=C_{\eta_0}$, labeling by $\nu_0$ the class of $x_0=0$ and by $\nu_1$ the class of $x_1=1$; the nodes $x_0,x_1$ are adjacent, so they lie in opposite classes.  The degree-zero cancellation makes the class masses equal, and the higher cancellations make all moments through degree $K$ equal.  In units of $\abs{w_0}$, each class has mass in $[1,\,e^{A_*/C}+A_*/C]$ by \cref{lem:quadratic-tail}, and the tail mass outside the distinguished atoms is at most $A_*/C$.  Choosing $C_{\eta_0}$ large therefore gives $\nu_0(\{0\})\ge1-{\eta_0}$ and $\nu_1(\{1\})\ge1-{\eta_0}$.

The common mean $\mu$ satisfies $\mu\ge1\cdot\nu_1(\{1\})\ge1/2$ and, by \eqref{eq:R1-new} and \eqref{eq:weighted-new}, $\mu\le e^{A_*/C}+A_*\le C_\mu$ for a universal $C_\mu$.  Rescaling both measures by $x\mapsto x/\mu$ preserves the moment identities in matched form and the atom at zero, sends the atom at $1$ into $[1/C_\mu,2]=:[c_0',c_1']$, and keeps the support inside $[0,2C_{\eta_0} K^2]$.
\end{proof}

\subsubsection{Equal-weight correction in the Chebyshev basis}

The measures of \cref{lem:continuous-twins} have unequal atom weights.  The next lemma converts them to equal-weight empirical multisets of size $r$, preserving the first $K$ moments exactly.  The bounded Chebyshev basis gives a size requirement polynomial in $K$.

\begin{lemma}[Equal-weight moment matching]\label{lem:cheb-correction}
Let $K\ge2$, $B\ge2$, and let $\bar\nu_0,\bar\nu_1$ be probability measures supported on at most $K+2$ points of $[0,B]$ with equal moments through degree $K$.  Fix $0<\eta<1/10$.  There is a universal constant $C_{\mathrm{pc}}$ such that, whenever
\begin{equation}\label{eq:correction-size}
r\ge C_{\mathrm{pc}}\,\eta^{-1}BK^5\log^2(K+2),
\end{equation}
there are multisets $X_0,X_1\subset[0,B]$ of cardinality $r$ whose empirical moments agree exactly through degree $K$.  In each $X_b$, at most $(\eta/B)\,r$ points form $K$ groups of control points; every other point sits at an atom of $\bar\nu_b$, with the count at each atom within $K+2$ of its target mass; and only the groups of control points of $X_1$ are displaced.
\end{lemma}

\begin{proof}
Set
\[
\theta=\frac\eta B,\qquad
N=\left\lfloor\frac{\theta r}{K}\right\rfloor,\qquad
R=r-KN,\qquad
w=\frac Nr,
\]
so that \eqref{eq:correction-size} gives $N\ge1$ and $\theta/(2K)\le w\le\theta/K$.  For
\[
\vartheta_j=\frac{j\pi}{K+1},
\qquad
y_j=\cos\vartheta_j,
\qquad
t_j=\frac B2(1+y_j),
\qquad 1\le j\le K,
\]
place $N$ control points at every $t_j$ in each multiset.  Allocate the remaining $R$ points to the atoms of $\bar\nu_b$ by flooring the target counts $R\,\bar\nu_b(\{z\})$ and distributing the leftover points among the largest fractional parts; each count is then within $1$ of its target, hence within $K+2$ of $r(1-\theta)\bar\nu_b(\{z\})$.

Work in the bounded basis
\[
\phi_k(x)=\mathsf{T}_k\!\left(\frac{2x}B-1\right),
\qquad 0\le k\le K,
\]
where $\mathsf{T}_k$ is the Chebyshev polynomial of the first kind; equality of the empirical $\phi_k$-moments for $0\le k\le K$ is equivalent to equality of the monomial moments through degree $K$.  The control contributions and the ideal base contributions cancel between $b=0$ and $b=1$, and $\abs{\phi_k}\le1$ on $[0,B]$, so the difference $e_k$ of the rounded empirical $\phi_k$-moments obeys
\begin{equation}\label{eq:cheb-rounding}
\abs{e_k}\le\frac{2(K+2)}r,
\qquad 1\le k\le K .
\end{equation}

Move the $j$th group of control points of $X_1$ from $t_j$ to $\tfrac B2(1+y_j+\delta_j)$.  In the normalized coordinate the moment equations read
\[
G_k(\delta)
=w\sum_{j=1}^K\bigl(\mathsf{T}_k(y_j+\delta_j)-\mathsf{T}_k(y_j)\bigr)
=-e_k,
\qquad 1\le k\le K .
\]
Writing $\mathsf{U}_{k-1}$ for the Chebyshev polynomial of the second kind, the Jacobian at the origin is
\[
A_{k,j}=w\,k\,\mathsf{U}_{k-1}(y_j)=w\,k\,\frac{\sin(k\vartheta_j)}{\sin\vartheta_j}.
\]
With $S_{k,j}=\sin(k\vartheta_j)$, the discrete sine transform satisfies $SS^{\mathsf T}=\tfrac{K+1}2I$, so
\[
\begin{aligned}
A^{-1}
&=\frac1w\,D_{\sin\vartheta}\,\frac2{K+1}\,S^{\mathsf T}D_{1/k},\\
\norm{A^{-1}}_{\infty\to\infty}
&\le\frac1w\cdot\frac2{K+1}\max_j\sum_{k=1}^K\frac{\abs{\sin(k\vartheta_j)}}k\\
&\le\frac{C\log(K+2)}{wK}
\le\frac{CB\log(K+2)}\eta .
\end{aligned}
\]
Combining with \eqref{eq:cheb-rounding},
\begin{equation}\label{eq:cheb-linear}
L_0:=\norm{A^{-1}e}_\infty
\le\frac{CBK\log(K+2)}{\eta r}.
\end{equation}

By the Markov brothers' inequality, $\max_{[-1,1]}\abs{\mathsf{T}_k''}\le k^4$.  Hence, on the $\ell_\infty$ ball of radius $\varrho$ around the origin, provided the shifted nodes stay in $[-1,1]$,
\[
\norm{DG(\delta)-A}_{\infty\to\infty}
\le wK\cdot K^4\varrho,
\qquad\text{so}\qquad
\norm{A^{-1}\bigl(DG(\delta)-A\bigr)}_{\infty\to\infty}
\le CK^4\log(K+2)\,\varrho .
\]
Choose $\varrho=c/(K^4\log(K+2))$ with $c$ small enough that the right-hand side is at most $1/2$.  Since $1-y_1$ and $1+y_K$ are at least $2/(K+1)^2$ and $\varrho\le1/(K+1)^2$, the ball stays inside $[-1,1]^K$.  The size condition \eqref{eq:correction-size} and \eqref{eq:cheb-linear} give $L_0\le\varrho/2$, so the Newton map
$\mathcal T(\delta)=-A^{-1}e+\delta-A^{-1}G(\delta)$
is a $\tfrac12$-Lipschitz self-map of the ball; its fixed point solves the moment equations and satisfies $\norm\delta_\infty\le2L_0\le\varrho$.  The adjusted control points remain in $[0,B]$, and equality of all $\phi_k$-moments, $0\le k\le K$, gives equality of the monomial moments through degree $K$.
\end{proof}

\begin{proof}[Proof of \cref{prop:moment-twins}]
Apply \cref{lem:continuous-twins} with parameter ${\eta_0}/4$ and let $B=2C_{{\eta_0}/4}K^2$.  Apply \cref{lem:cheb-correction} with $\eta={\eta_0}/4$; the size condition \eqref{eq:correction-size} follows from \eqref{eq:twin-size} after increasing $C_{\eta_0}$.

In $X_0$, the control points occupy a fraction at most $\eta/B\le{\eta_0}/8$, and the rounded base points away from zero occupy a fraction at most ${\eta_0}/4+O(K/r)$.  Thus at most ${\eta_0} r$ points are nonzero for large $r$.  In $X_1$, the distinguished atom of $\bar\nu_1$ is a base atom, hence never displaced; it carries at least
$(1-\eta)(1-{\eta_0}/4)r-(K+2)\ge r/2$
points at one common value in $[c_0',c_1']$.

The two empirical means are equal, and their common value $M$ differs from the base mean $1$ by at most the control contribution $\theta\cdot B=\eta$ plus a rounding term $O(K^2B/r)=o(1)$; hence $M\in[1/2,2]$ for large $r$.  Dividing every point by $rM$ yields vectors $a^{(b)}$ with $\sum_ia_i^{(b)}=1$, preserves \eqref{eq:moment-match} and the zero entries, sends the distinguished value into $[c_0'/2,\,2c_1']\cdot r^{-1}=:[c_0/r,c_1/r]$, and bounds every entry by $2B/r$.  Increasing $C_{\eta_0}$ absorbs the constants.
\end{proof}

\subsubsection{Quadratic spread is necessary}

\begin{proposition}[Optimality of the $K^2$ spread]\label{prop:spread-optimal}
Fix $0<c_0<c_1$ and $0<\beta\le1$, and let $0<{\eta_0}\le\beta/4$ and $K\ge1$.  Suppose two probability measures $\nu_0,\nu_1$ on $[0,B]$ have identical moments through degree $K$, with
\[
\nu_0(\{0\})\ge1-{\eta_0},
\qquad
\nu_1([c_0,c_1])\ge\beta .
\]
Then $B\ge cK^2$ for a constant $c=c(c_0,\beta)>0$.
\end{proposition}

\begin{proof}
Matching the first moments gives ${\eta_0} B\ge\beta c_0$, hence $B\ge4c_0$.  Let
\[
Q_K(x)=
\mathsf{T}_K\!\left(\frac{2x-(B+c_0)}{B-c_0}\right)
\Big/\,
\mathsf{T}_K\!\left(-\frac{B+c_0}{B-c_0}\right).
\]
Then $Q_K(0)=1$; $\abs{Q_K}\le1$ on $[0,c_0]$, because $\abs{\mathsf{T}_K(x)}$ increases with $|x|$ for $x\le-1$; and on $[c_0,B]$ the numerator's argument lies in $[-1,1]$, so
\[
\abs{Q_K(x)}\le\delta_K:=\abs{\mathsf{T}_K(z_*)}^{-1},
\qquad
z_*=\frac{B+c_0}{B-c_0}.
\]
Hence
\[
\E_{\nu_0}Q_K\ge1-2{\eta_0},
\qquad
\E_{\nu_1}Q_K\le1-\beta(1-\delta_K),
\]
and moment matching forces the two expectations to be equal, so $\delta_K\ge1-2{\eta_0}/\beta\ge1/2$.  Writing $z_*=1+u$ with $u=2c_0/(B-c_0)\le1$, the identity $\mathsf{T}_K(\cosh t)=\cosh(Kt)$ gives
$\cosh\bigl(K\operatorname{arcosh}(1+u)\bigr)\le2$, hence $K\operatorname{arcosh}(1+u)\le2$.  Since $\cosh\sqrt u\le1+u$ for $0\le u\le1$, we get $\operatorname{arcosh}(1+u)\ge\sqrt u$, so $K\sqrt{2c_0/B}\le K\sqrt u\cdot\sqrt2\le2\sqrt2$ and $B\ge c_0K^2/4$.
\end{proof}

\subsection{A size-biased doubly correlated Wishart model}\label{sec:wishart}

Fix vectors $a^{(0)},a^{(1)}$ supplied by \cref{prop:moment-twins}, and let
\[
D_b=\operatorname{diag}(a_1^{(b)},\ldots,a_r^{(b)}).
\]
The support parameter ${\eta_0}$ will be fixed in \cref{sec:separation}; the moderate-block constants $\beta,c_0,c_1$ remain universal as ${\eta_0}$ decreases.
Let $Z$ be an $r\times r$ standard complex Ginibre matrix: its entries are independent complex Gaussians with density $\pi^{-1}e^{-|z|^2}$.  Let $U,V$ be independent Haar unitaries, independent of $Z$, and define
\begin{equation}\label{eq:G-def}
G_b=UD_b^{1/2}ZD_b^{1/2}V^*,
\qquad
R_b=\norm{G_b}_F^2.
\end{equation}
Let $Q_b$ denote the law of $G_b$.

For an integer $m\ge1$, define the size-biased prior
\begin{equation}\label{eq:tilted-prior}
\frac{d\Pi_b^{(m)}}{dQ_b}(G)
=
\frac{\norm G_F^{2m}}{\E_{Q_b}\norm G_F^{2m}}.
\end{equation}
For $G\ne0$, put
\begin{equation}\label{eq:psi-G}
\ket{\psi_G}=\frac{\operatorname{vec}(G)}{\norm G_F}
\in\C^r\otimes\C^r.
\end{equation}
The corresponding $m$-copy average pure state is
\begin{equation}\label{eq:Bayes-pure}
\overline\Psi_b^{(m)}
=
\E_{G\sim\Pi_b^{(m)}}\proj{\psi_G}^{\otimes m}.
\end{equation}
The tilt gives the cancellation
\begin{equation}\label{eq:size-bias-cancel}
\overline\Psi_b^{(m)}
=
\frac{\E_{Q_b}\bigl[\proj{\operatorname{vec}(G)}^{\otimes m}\bigr]}
{\E_{Q_b}R_b^m}.
\end{equation}

\subsubsection{Doubly correlated Wishart moments}

\begin{lemma}[Doubly correlated Wishart Schur moment]\label{lem:wishart-schur}
Let $A=\operatorname{diag}(a_1,\ldots,a_r)$ and $B=\operatorname{diag}(b_1,\ldots,b_r)$ be nonnegative diagonal matrices, let $Z$ be standard complex Ginibre, and set
\[
W=A^{1/2}ZBZ^*A^{1/2}.
\]
For every partition $\lambda\vdash m$,
\begin{equation}\label{eq:wishart-schur}
\E s_\lambda(W)
=\frac{m!}{f^\lambda}s_\lambda(a)s_\lambda(b).
\end{equation}
\end{lemma}

\begin{proof}
Let $P_\pi$ denote the tensor-factor permutation indexed by $\pi\in S_m$.  Wick's formula gives the operator identity
\begin{equation}\label{eq:wick-tensor}
\E(ZBZ^*)^{\otimes m}
=\sum_{\pi\in S_m}p_\pi(b)P_\pi.
\end{equation}
\Cref{app:wick} derives \eqref{eq:wick-tensor} entrywise and fixes the permutation convention.
Let $\mathsf P_\lambda$ be the central Schur--Weyl projector onto the $\lambda$-isotypic subspace of $(\C^r)^{\otimes m}$.  On that subspace, $A^{\otimes m}$ acts as the polynomial $GL_r$ representation indexed by $\lambda$, while $P_\pi$ acts on the multiplicity space.  Therefore
\begin{equation}\label{eq:trace-block}
\tr\bigl(\mathsf P_\lambda A^{\otimes m}P_\pi\bigr)
=s_\lambda(a)\chi^\lambda(\pi).
\end{equation}
Also
\[
s_\lambda(W)
=\frac1{f^\lambda}\tr\bigl(\mathsf P_\lambda W^{\otimes m}\bigr).
\]
Combining this with \eqref{eq:wick-tensor} and \eqref{eq:trace-block},
\begin{align*}
\E s_\lambda(W)
&=\frac{s_\lambda(a)}{f^\lambda}
  \sum_{\pi\in S_m}\chi^\lambda(\pi)p_\pi(b)\\
&=\frac{m!}{f^\lambda}s_\lambda(a)s_\lambda(b),
\end{align*}
where the last equality is the Frobenius formula \eqref{eq:frobenius}.
\end{proof}

\subsubsection{Schur law of the state averaged over the prior}

\begin{lemma}[Explicit local unitary twirl]\label{lem:twirl}
Let $p=(p_1,\ldots,p_r)$ be the Schmidt spectrum of a bipartite pure state $\ket{\psi_p}\in\C^r\otimes\C^r$.  For every $\lambda\vdash m$ with $\ell(\lambda)\le r$, choose a real orthogonal model of the symmetric-group irrep $\mathcal P_\lambda$, identify $\mathcal P_\lambda^B$ with the dual of $\mathcal P_\lambda^A$, and choose the resulting paired orthonormal bases.  Define
\begin{equation}\label{eq:Phi-Plambda}
\ket{\Phi_\lambda^{\mathcal P}}
=\frac1{\sqrt{f^\lambda}}
\sum_{t=1}^{f^\lambda}\ket t_{\mathcal P_\lambda^A}\ket t_{\mathcal P_\lambda^B}.
\end{equation}
Define the density operator
\begin{equation}\label{eq:Omega-lambda-explicit}
\Omega_\lambda
=
\frac{I_{\mathcal Q_\lambda^A}}{d_\lambda}
\otimes
\frac{I_{\mathcal Q_\lambda^B}}{d_\lambda}
\otimes
\proj{\Phi_\lambda^{\mathcal P}},
\end{equation}
embedded in the $(\lambda,\lambda)$ Schur block of the two subsystems.  Then
\begin{equation}\label{eq:twirl-decomp}
\int (U\otimes V)^{\otimes m}
\proj{\psi_p}^{\otimes m}
(U^*\otimes V^*)^{\otimes m}\,dU\,dV
=
\bigoplus_{\substack{\lambda\vdash m\\\ell(\lambda)\le r}}
f^\lambda s_\lambda(p)\,\Omega_\lambda.
\end{equation}
In particular, the conditional state given the common Schur label $\lambda$ is independent of $p$.
\end{lemma}

\begin{proof}
Let $\mathcal S_A,\mathcal S_B$ be Schur--Weyl unitaries as in \eqref{eq:schur-weyl-decomp}.  The vector $\ket{\psi_p}^{\otimes m}$ is invariant under simultaneous permutations of its $m$ bipartite tensor factors.  Since every symmetric-group irreducible is self-dual, Schur's lemma implies that its component in the local label pair $(\lambda,\mu)$ vanishes unless $\lambda=\mu$, and that the multiplicity-space component in the $(\lambda,\lambda)$ block is proportional to the unique invariant vector \eqref{eq:Phi-Plambda}.  Thus there are unit vectors
\[
\ket{\xi_{\lambda,p}}
\in\mathcal Q_\lambda^A\otimes\mathcal Q_\lambda^B
\]
such that
\begin{equation}\label{eq:pure-schur-decomp}
(\mathcal S_A\otimes\mathcal S_B)\ket{\psi_p}^{\otimes m}
=
\bigoplus_{\substack{\lambda\vdash m\\\ell(\lambda)\le r}}
\sqrt{f^\lambda s_\lambda(p)}\,
\ket{\xi_{\lambda,p}}\otimes\ket{\Phi_\lambda^{\mathcal P}}.
\end{equation}
The squared norm of the $\lambda$ component is the probability of the first subsystem's Schur label,
\[
\tr(\mathsf P_\lambda\rho_p^{\otimes m})
=f^\lambda s_\lambda(p),
\]
where $\rho_p$ is the reduced state.  The weights sum to one by \eqref{eq:p1-schur}.

Under $U^{\otimes m}$, the $\lambda$ block transforms irreducibly on $\mathcal Q_\lambda^A$ and trivially on $\mathcal P_\lambda^A$; similarly for the second subsystem.  Haar orthogonality kills cross terms between inequivalent labels.  On a fixed diagonal block, independent twirling on the two irreducible representation spaces sends every operator $X$ on $\mathcal Q_\lambda^A\otimes\mathcal Q_\lambda^B$ to
\[
\frac{\tr X}{d_\lambda^2}
I_{\mathcal Q_\lambda^A}\otimes I_{\mathcal Q_\lambda^B}.
\]
Applying this to \eqref{eq:pure-schur-decomp} gives \eqref{eq:twirl-decomp} and \eqref{eq:Omega-lambda-explicit}.
\end{proof}

\begin{proposition}[Schur decomposition of the state averaged over the prior]\label{prop:schur-bayes}
For the prior \eqref{eq:tilted-prior},
\begin{equation}\label{eq:schur-bayes-decomp}
\overline\Psi_b^{(m)}
=
\bigoplus_{\lambda\vdash m}P_b^{(m)}(\lambda)\,\Omega_\lambda,
\end{equation}
where the blocks $\Omega_\lambda$ are common to $b=0,1$ and
\begin{equation}\label{eq:schur-bayes-law}
P_b^{(m)}(\lambda)
=
\frac{s_\lambda(a^{(b)})^2}
{\sum_{\mu\vdash m}s_\mu(a^{(b)})^2}.
\end{equation}
Consequently,
\begin{equation}\label{eq:trace-TV}
d_{\mathrm{tr}}(\overline\Psi_0^{(m)},\overline\Psi_1^{(m)})
=\TV(P_0^{(m)},P_1^{(m)}).
\end{equation}
\end{proposition}

\begin{proof}
By the vectorization convention \eqref{eq:vec-convention}, the Haar factors in \eqref{eq:G-def} act on $\operatorname{vec}(G)$ as the independent local Haar rotations $U\otimes\overline V$; the radial size bias is independent of both.  For a fixed $G$, let $W=GG^*$ and $R_G=\tr W$.  Its normalized Schmidt spectrum is the spectrum of $W/R_G$.  Since $s_\lambda$ is homogeneous of degree $m$, multiplying the label weight in \eqref{eq:twirl-decomp} by the size-bias factor $R_G^m$ gives
\[
R_G^m f^\lambda s_\lambda(W/R_G)=f^\lambda s_\lambda(W).
\]
Applying \cref{lem:wishart-schur} with $A=B=D_b$ gives
\[
f^\lambda\E s_\lambda(W)
=m!s_\lambda(a^{(b)})^2.
\]
Finally, \eqref{eq:p1-schur} implies
\[
\E R_G^m
=\sum_{\lambda\vdash m}f^\lambda\E s_\lambda(W)
=m!\sum_{\lambda\vdash m}s_\lambda(a^{(b)})^2.
\]
This proves \eqref{eq:schur-bayes-decomp}--\eqref{eq:schur-bayes-law}.  Since both hypotheses have the same normalized block state $\Omega_\lambda$, \eqref{eq:block-trace-norm} gives
\[
\frac12\norm{\overline\Psi_0^{(m)}-\overline\Psi_1^{(m)}}_1
=\frac12\sum_\lambda\abs{P_0^{(m)}(\lambda)-P_1^{(m)}(\lambda)},
\]
which is \eqref{eq:trace-TV}.
\end{proof}

\subsection{Indistinguishability}\label{sec:indist}

For $b,c\in\{0,1\}$, define
\begin{equation}\label{eq:Cbc}
C_{bc}^{(m)}
=\sum_{\lambda\vdash m}s_\lambda(a^{(b)})s_\lambda(a^{(c)}),
\qquad
Z_{bc}^{(m)}=m!C_{bc}^{(m)}.
\end{equation}
The Hellinger affinity of the two laws in \eqref{eq:schur-bayes-law} is
\begin{equation}\label{eq:affinity-C}
\mathcal A_m
=\frac{C_{01}^{(m)}}{\sqrt{C_{00}^{(m)}C_{11}^{(m)}}}.
\end{equation}

Let
\[
u_{b,k}=p_k(a^{(b)})=\sum_i(a_i^{(b)})^k.
\]
The Cauchy identity \eqref{eq:cauchy} and the exponential formula give
\begin{align}
\sum_{m\ge0}C_{bc}^{(m)}z^m
&=\exp\!\left(\sum_{k\ge1}\frac{z^k}{k}u_{b,k}u_{c,k}\right),
\label{eq:C-generating}\\
Z_{bc}^{(m)}
&=\sum_{\pi\in S_m}
\prod_{d\in\operatorname{Cycles}(\pi)}
 u_{b,|d|}u_{c,|d|}.
\label{eq:permutation-sum}
\end{align}

\subsubsection{The long-cycle estimate}

For $b\in\{0,1\}$, normalize the nonnegative weights in $Z_{bb}^{(m)}$ to a probability law on $S_m$:
\begin{equation}\label{eq:weighted-perm-law}
\Prb_b(\pi)
=\frac1{Z_{bb}^{(m)}}
\prod_{d\in\operatorname{Cycles}(\pi)}u_{b,|d|}^2.
\end{equation}
Let $N_k(\pi)$ be the number of $k$-cycles.

\begin{lemma}\label{lem:long-cycles}
Suppose $u_{b,1}=1$ and $\max_i a_i^{(b)}\le B/r$.  Then
\begin{equation}\label{eq:ENk}
\E_bN_k
\le
\frac{m^k}{k}\left(\frac Br\right)^{2k-2}.
\end{equation}
Consequently, if
\[
\theta=\frac{mB^2}{r^2}<1,
\]
then
\begin{equation}\label{eq:long-cycle-prob}
\delta_b
:=\Prb_b[\text{some cycle has length greater than }K]
\le
\frac{m\theta^K}{(K+1)(1-\theta)}.
\end{equation}
\end{lemma}

\begin{proof}
Selecting the labels of a distinguished $k$-cycle and arranging them cyclically gives
\begin{equation}\label{eq:exact-ENk}
\E_bN_k
=
\frac{(m)_k}{k}u_{b,k}^2
\frac{Z_{bb}^{(m-k)}}{Z_{bb}^{(m)}}.
\end{equation}
Because $u_{b,1}=1$, adjoining $k$ fixed points embeds each weighted permutation on $m-k$ labels into one on $m$ labels without changing its weight.  Hence $Z_{bb}^{(m)}\ge Z_{bb}^{(m-k)}$.  Moreover,
\[
u_{b,k}\le\left(\max_i a_i^{(b)}\right)^{k-1}\sum_i a_i^{(b)}
\le(B/r)^{k-1}.
\]
This proves \eqref{eq:ENk}.  A union bound gives
\begin{align*}
\delta_b
&\le\sum_{k=K+1}^m\E_bN_k\\
&\le\sum_{k=K+1}^m
\frac{m^k}{k}\left(\frac Br\right)^{2k-2}\\
&=m\sum_{k=K+1}^m\frac{\theta^{k-1}}k
\le\frac{m\theta^K}{(K+1)(1-\theta)}.
\end{align*}
\end{proof}

\begin{proposition}[Cauchy cancellation bound]\label{prop:cauchy-cancel}
Assume
\[
u_{0,k}=u_{1,k},
\qquad 1\le k\le K,
\]
and $\max_{b,i}a_i^{(b)}\le B/r$.  If $\theta=mB^2/r^2<1$, then
\begin{equation}\label{eq:trace-bound-final}
d_{\mathrm{tr}}(\overline\Psi_0^{(m)},\overline\Psi_1^{(m)})
\le
\left[
\frac{2m\theta^K}{(K+1)(1-\theta)}
\right]^{1/2}.
\end{equation}
\end{proposition}

\begin{proof}
Let $Z_*^{(m)}$ be the total weight in \eqref{eq:permutation-sum} over permutations all of whose cycle lengths are at most $K$.  Moment matching makes this weight identical in $Z_{00}^{(m)},Z_{01}^{(m)},Z_{11}^{(m)}$.  By definition of $\delta_b$,
\[
Z_*^{(m)}=(1-\delta_b)Z_{bb}^{(m)},
\qquad b=0,1.
\]
All cross weights are nonnegative, so $Z_{01}^{(m)}\ge Z_*^{(m)}$.  Therefore
\begin{align*}
\mathcal A_m
&=\frac{Z_{01}^{(m)}}{\sqrt{Z_{00}^{(m)}Z_{11}^{(m)}}}\\
&\ge\sqrt{(1-\delta_0)(1-\delta_1)}.
\end{align*}
Using \eqref{eq:tv-affinity}, \eqref{eq:trace-TV}, and \cref{lem:long-cycles},
\[
d_{\mathrm{tr}}
\le\sqrt{1-\mathcal A_m^2}
\le\sqrt{\delta_0+\delta_1},
\]
which is \eqref{eq:trace-bound-final}.
\end{proof}

\subsubsection{Choosing the base copy count}

Fix
\begin{equation}\label{eq:K-choice}
\theta_0=\frac12,
\qquad
K=\lfloor5\log(r+2)\rfloor,
\end{equation}
let $B=C_{\eta_0} K^2$ be the entry bound from \cref{prop:moment-twins}, and put
\begin{equation}\label{eq:m-choice}
m=\left\lfloor\theta_0\frac{r^2}{B^2}\right\rfloor .
\end{equation}
The size condition \eqref{eq:twin-size} holds for all large $r$, and $\theta=mB^2/r^2\le\theta_0$.  Since $5\log 2>3$,
\[
m\,\theta_0^K\le r^2\,2^{-K}\le 8\,r^{-1},
\]
so the right-hand side of \eqref{eq:trace-bound-final} is at most $4r^{-1/2}$.  At the same time,
\begin{equation}\label{eq:m-polylog}
m\ge c\,\frac{r^2}{\log^4(r+2)}
\end{equation}
for a constant $c=c({\eta_0})>0$.

\begin{corollary}[Base indistinguishability]\label{cor:base-indist}
For the priors $\Pi_b^{(m)}$ constructed above, with $m$ given by \eqref{eq:m-choice},
\[
d_{\mathrm{tr}}(\overline\Psi_0^{(m)},\overline\Psi_1^{(m)})\le4r^{-1/2},
\qquad
m\ge c\,\frac{r^2}{\log^4(r+2)} .
\]
\end{corollary}

\subsection{Separation in nuclear norm}\label{sec:separation}

For $G\ne0$, let
\[
\rho_G=\frac{GG^*}{\norm G_F^2}.
\]
Against the maximally mixed reference on $\C^r$,
\begin{equation}\label{eq:T-def}
T(G)
:=\Fid\!\left(\rho_G,\frac{I_r}{r}\right)
=\frac{\norm G_1}{\sqrt r\norm G_F}.
\end{equation}

\subsubsection{The low-rank side}

By \eqref{eq:low-support}, $D_0$ has rank at most ${\eta_0} r$.  Hence every $G$ in the support of $Q_0$, and therefore of $\Pi_0^{(m)}$, has rank at most ${\eta_0} r$.  Cauchy--Schwarz for singular values gives
\begin{equation}\label{eq:low-T}
T(G)\le\sqrt{\eta_0}
\qquad\Pi_0^{(m)}\text{-a.s.}
\end{equation}

\subsubsection{A Ginibre nuclear norm bound}

\begin{lemma}[Ginibre nuclear norm]\label{lem:ginibre-nuclear}
There are universal constants $c_g,c'_g>0$ such that, for every $s\ge1$ and every standard complex Ginibre matrix $Z_s\in\C^{s\times s}$,
\begin{equation}\label{eq:ginibre-tail}
\Prb\bigl[\norm{Z_s}_1\le c_gs^{3/2}\bigr]
\le e^{-c'_gs^2}.
\end{equation}
\end{lemma}

\begin{proof}
Schatten interpolation gives
\begin{equation}\label{eq:schatten-interp}
\norm Z_1
\ge\frac{\norm Z_F^3}{\norm Z_4^2},
\qquad
\norm Z_4^4=\tr((ZZ^*)^2).
\end{equation}
This follows from the log-convexity relation
$\norm Z_F\le\norm Z_1^{1/3}\norm Z_4^{2/3}$ for the Schatten $1,2,4$ norms.

The random variable $\norm Z_F^2$ is a sum of $s^2$ independent mean-one exponentials.  A Chernoff bound gives
\begin{equation}\label{eq:ginibre-F-tail}
\Prb[\norm Z_F^2<s^2/2]\le e^{-c s^2}.
\end{equation}
An index-level Wick calculation gives
\begin{equation}\label{eq:ginibre-fourth-moment}
\E\tr((ZZ^*)^2)=2s^3.
\end{equation}
Hence Markov's inequality implies
\[
\Prb[\tr((ZZ^*)^2)>8s^3]\le\frac14.
\]
For all sufficiently large $s$, the events
$\norm Z_F^2\ge s^2/2$ and $\tr((ZZ^*)^2)\le8s^3$ therefore intersect with probability at least $1/2$.  On their intersection, \eqref{eq:schatten-interp} gives
\[
\norm Z_1
\ge
\frac{(s^2/2)^{3/2}}{(8s^3)^{1/2}}
=\frac18s^{3/2}.
\]
Decreasing the constant to cover the remaining values of $s$ gives
\begin{equation}\label{eq:ginibre-nuclear-mean}
\E\norm Z_1\ge c_*s^{3/2}
\end{equation}
for a universal $c_*>0$.

The map $Z\mapsto\norm Z_1$ is $\sqrt s$-Lipschitz with respect to Frobenius norm, because
\[
\abs{\norm X_1-\norm Y_1}
\le\norm{X-Y}_1
\le\sqrt s\norm{X-Y}_F.
\]
Apply \cref{lem:gaussian-concentration} with
$t=(c_*/2)s^{3/2}$.  It gives
\[
\Prb\!\left[\norm Z_1\le\frac{c_*}{2}s^{3/2}\right]
\le e^{-c's^2}.
\]
Taking $c_g=c_*/2$ proves the lemma.
\end{proof}

For $b=1$, let $J$ be the set from \eqref{eq:moderate-block}; then $s=|J|\ge\beta r$ and $D_{1,J}\succeq(c_0/r)I_J$.  Put $H=D_1^{1/2}ZD_1^{1/2}$.  Haar multiplication preserves singular values, compression is contractive in nuclear norm, and multiplication by a matrix with least singular value $q$ decreases the nuclear norm by at most a factor $q$.  Therefore
\begin{align}
\norm{G}_1=\norm H_1
&\ge\norm{P_JHP_J}_1\notag\\
&=\norm{D_{1,J}^{1/2}Z_{JJ}D_{1,J}^{1/2}}_1\notag\\
&\ge\frac{c_0}{r}\norm{Z_{JJ}}_1.
\label{eq:block-nuclear}
\end{align}
By \cref{lem:ginibre-nuclear}, after decreasing constants if necessary one may take
\[
\kappa=c_g c_0\beta^{3/2}>0,
\]
and then
\begin{equation}\label{eq:nuclear-high-base}
\Prb_{Q_1}[\norm G_1<\kappa\sqrt r]
\le e^{-c r^2}.
\end{equation}

\subsubsection{Concentration under the size-biased prior}

Write
\begin{equation}\label{eq:S-weighted-sum}
R_G=\norm G_F^2
=\sum_{i,j}a_i^{(1)}a_j^{(1)}X_{ij},
\end{equation}
where the $X_{ij}$ are independent mean-one exponentials.  The coefficients sum to one and satisfy
\[
\zeta:=\max_{i,j}a_i^{(1)}a_j^{(1)}\le\frac{B^2}{r^2}.
\]

\begin{lemma}\label{lem:weighted-exp}
For every integer $\ell\ge1$,
\begin{equation}\label{eq:S-moment}
1\le\E R_G^\ell
\le\prod_{j=0}^{\ell-1}(1+j\zeta)
\le\exp\!\left(\frac{\zeta \ell^2}{2}\right).
\end{equation}
\end{lemma}

\begin{proof}
The lower bound is Jensen's inequality.  For the upper bound,
\begin{align*}
\log\E e^{tR_G}
&=\sum_{k\ge1}\frac{t^k}{k}
  \sum_{i,j}(a_i^{(1)}a_j^{(1)})^k\\
&\preceq\sum_{k\ge1}\frac{t^k}{k}\zeta^{k-1}
=-\frac1{\zeta}\log(1-\zeta t),
\end{align*}
where $\preceq$ denotes coefficientwise inequality of formal power series.  Exponentiation preserves coefficientwise order for series with nonnegative coefficients.  Comparing coefficients in the two exponential generating functions gives
\[
\E R_G^\ell\le\prod_{j=0}^{\ell-1}(1+j\zeta),
\]
and $\log(1+x)\le x$ gives the final inequality.
\end{proof}

For $m$ in \eqref{eq:m-choice},
\begin{equation}\label{eq:mcstar}
\zeta m\le\frac{mB^2}{r^2}\le\theta_0=\frac12 .
\end{equation}
Under the $R_G^m$-tilted law, since $2\zeta m^2\le m$,
\begin{align}
\Prb_{\Pi_1^{(m)}}[R_G>4]
&\le\frac{4^{-m}\E R_G^{2m}}{\E R_G^m}
\le\exp(2\zeta m^2-m\log4)
\le e^{-m/3}.
\label{eq:S-upper-tilted}
\end{align}
By Cauchy--Schwarz, \eqref{eq:S-moment}, and \eqref{eq:nuclear-high-base},
\begin{align}
\Prb_{\Pi_1^{(m)}}[\norm G_1<\kappa\sqrt r]
&\le
\frac{(\E R_G^{2m})^{1/2}
\Prb_{Q_1}[\norm G_1<\kappa\sqrt r]^{1/2}}
{\E R_G^m}
\le\exp\!\left(\frac m2-cr^2\right)
=e^{-\Omega(r^2)},
\label{eq:nuclear-tilted}
\end{align}
because $m=o(r^2)$ by \eqref{eq:K-choice} and \eqref{eq:m-choice}.
On the intersection of the events $R_G\le4$ and $\norm G_1\ge\kappa\sqrt r$, one has $T(G)\ge\kappa/2$.  Set
\[
\gamma=\kappa/2 .
\]
Combining \eqref{eq:S-upper-tilted} and \eqref{eq:nuclear-tilted} gives
\begin{equation}\label{eq:high-T-prob}
\Prb_{\Pi_1^{(m)}}[T(G)<\gamma]
\le e^{-m/3}+e^{-cr^2}=o(1).
\end{equation}

Choose the constant ${\eta_0}$ in \cref{prop:moment-twins} so small that
\begin{equation}\label{eq:support-fraction}
\sqrt{\eta_0}\le\gamma/4.
\end{equation}
This is possible without changing the lower bounds on $\beta$ and $c_0$, by \cref{prop:moment-twins}.
Condition $\Pi_1^{(m)}$ on $T(G)\ge\gamma$, and denote the result by $\widetilde\Pi_1^{(m)}$.  If $p_r$ is the discarded probability, then for every $k\ge1$ the original and conditional $k$-copy averages satisfy
\[
\E_{\Pi_1^{(m)}}\proj{\psi_G}^{\otimes k}
=(1-p_r)\E_{\widetilde\Pi_1^{(m)}}\proj{\psi_G}^{\otimes k}
+p_r\Xi^{(k)}
\]
for some density operator $\Xi^{(k)}$.  Hence
$d_{\mathrm{tr}}(\E_{\Pi_1^{(m)}}\proj{\psi_G}^{\otimes k},\E_{\widetilde\Pi_1^{(m)}}\proj{\psi_G}^{\otimes k})\le p_r=o(1)$.  Thus conditioning increases the bound in \cref{cor:base-indist} by at most $p_r$, and
\begin{equation}\label{eq:base-gap}
T(G)\le\gamma/4
\quad\Pi_0^{(m)}\text{-a.s.},
\qquad
T(G)\ge\gamma
\quad\widetilde\Pi_1^{(m)}\text{-a.s.}
\end{equation}

\begin{proof}[Proof of \cref{prop:flat-priors}]
Set $r=j$.  Choose the universal constant ${\eta_0}$ according to \eqref{eq:support-fraction}, and use $K,m$ from \eqref{eq:K-choice} and \eqref{eq:m-choice}.
Let $\Lambda_0$ be the law of $\tau=GG^*/\norm G_F^2$ under $\Pi_0^{(m)}$, and let $\Lambda_1$ be its law under $\widetilde\Pi_1^{(m)}$.
Then $T_j(\tau)=T(G)$, so \eqref{eq:base-gap} gives the two required intervals.
By \cref{cor:base-indist} and the conditioning bound,
the pure $m$-copy averages have trace distance at most
$4j^{-1/2}+e^{-m/3}+e^{-cj^2}$.
This is at most $1/100$ for all sufficiently large $j$.
Tracing out the purifying registers and then any $m-t$ copies proves the same bound for the mixed $t$-copy averages, for every $0\le t\le m$.
Finally \eqref{eq:m-polylog} gives $m\ge cj^2/\log^4(j+1)$.
\end{proof}

\section{Discussion}

The rescaled covariant measurement gives a quadratic bound in the smaller rank and a spectral upper bound that matches the polynomial accuracy exponent of the lower bound for $1<\alpha\le2$.  A full characterization for arbitrary reference spectra remains open.

For $2<\alpha<3$, \cref{cor:power-laws} leaves a polynomial gap.  At $\alpha=5/2$, for example, the lower and upper bounds are
$c\varepsilon^{-8/3}/\log^4(1/\varepsilon)$ and
$C\varepsilon^{-10/3}$ in the stated rank regime.
Improving either side would clarify which spectral features determine the complexity.
For a pure target with white noise, the bound
$O((d+\mu d^2)/\varepsilon^2)$ leaves a related question when $\mu<1/d$ and high precision is required.

Logarithmic factors also remain to be determined.
At fixed accuracy the maximally mixed reference has complexity
$\Theta_\varepsilon(r^2/\log^4r)$.
Does every rank-$r$ reference admit this rate?
For inverse-square spectra, the remaining gap in \eqref{eq:inverse-square} is likewise logarithmic.

The measurement in \cref{thm:unknown-rank} depends on a supplied bound on the unknown rank.
Adapting to approximate rank and determining the implementation cost of the estimator are further questions.
Extending the quadratic dependence on the smaller rank to two unknown input states remains the distinct problem identified by Lowe--Tan~\cite{LoweTan2026}.

\section*{AI usage disclosure}

AI systems from OpenAI and Anthropic assisted with mathematical exploration, including development of the rescaled estimator and its error analysis, proof checking, and manuscript preparation.  The authors are responsible for the mathematical claims and the final text.

\appendix

\section{Index-level derivation of the Wishart--Schur identity}\label{app:wick}

We derive \eqref{eq:wick-tensor} and \eqref{eq:wishart-schur} entrywise to fix the permutation and normalization conventions.

Let $B=\operatorname{diag}(b_1,\ldots,b_r)$ and set $Y=ZBZ^*$.  For multi-indices $\mathbf i=(i_1,\ldots,i_m)$ and $\mathbf j=(j_1,\ldots,j_m)$,
\begin{align}
\E\prod_{t=1}^mY_{i_tj_t}
&=
\sum_{\alpha_1,\ldots,\alpha_m}
\left(\prod_{t=1}^m b_{\alpha_t}\right)
\E\prod_{t=1}^m
Z_{i_t\alpha_t}\overline{Z_{j_t\alpha_t}}.
\label{eq:wick-index-start}
\end{align}
The complex Wick formula gives
\begin{equation}\label{eq:complex-wick-index}
\E\prod_{t=1}^m
Z_{i_t\alpha_t}\overline{Z_{j_t\alpha_t}}
=
\sum_{\pi\in S_m}
\prod_{t=1}^m
\delta_{i_t,j_{\pi(t)}}
\delta_{\alpha_t,\alpha_{\pi(t)}}.
\end{equation}
For a fixed $\pi$, summing the $\alpha$ indices along each cycle gives
\[
\prod_{c\in\operatorname{Cycles}(\pi)}
\sum_{a=1}^r b_a^{|c|}
=p_\pi(b).
\]
Thus
\begin{equation}\label{eq:wick-index-finish}
\E\prod_{t=1}^mY_{i_tj_t}
=
\sum_{\pi\in S_m}p_\pi(b)
\prod_{t=1}^m\delta_{i_t,j_{\pi(t)}}.
\end{equation}
If tensor permutations are defined by
\[
P_\pi(v_1\otimes\cdots\otimes v_m)
=v_{\pi^{-1}(1)}\otimes\cdots\otimes v_{\pi^{-1}(m)},
\]
then the matrix elements on the right of \eqref{eq:wick-index-finish} are those of $P_{\pi^{-1}}$.  Replacing $\pi$ by $\pi^{-1}$, which preserves cycle type, proves
\[
\E Y^{\otimes m}=\sum_{\pi\in S_m}p_\pi(b)P_\pi.
\]

Now let $W=A^{1/2}YA^{1/2}$.  Since $A^{\otimes m}$ commutes with every tensor permutation,
\[
\E W^{\otimes m}
=A^{\otimes m}\sum_{\pi\in S_m}p_\pi(b)P_\pi.
\]
On the Schur--Weyl block $\mathcal Q_\lambda\otimes\mathcal P_\lambda$,
$A^{\otimes m}$ acts as the polynomial representation $q_\lambda(A)$ on $\mathcal Q_\lambda$, and $P_\pi$ acts as the irreducible representation $\rho_\lambda(\pi)$ on $\mathcal P_\lambda$.  Therefore
\begin{equation}\label{eq:appendix-block-trace}
\tr(\mathsf P_\lambda A^{\otimes m}P_\pi)
=\tr q_\lambda(A)\,\tr\rho_\lambda(\pi)
=s_\lambda(a)\chi^\lambda(\pi).
\end{equation}
Likewise,
\[
\tr(\mathsf P_\lambda W^{\otimes m})
=f^\lambda s_\lambda(W).
\]
Consequently,
\begin{align*}
\E s_\lambda(W)
&=\frac1{f^\lambda}
\sum_{\pi\in S_m}p_\pi(b)
\tr(\mathsf P_\lambda A^{\otimes m}P_\pi)\\
&=\frac{s_\lambda(a)}{f^\lambda}
\sum_{\pi\in S_m}\chi^\lambda(\pi)p_\pi(b)\\
&=\frac{m!}{f^\lambda}s_\lambda(a)s_\lambda(b),
\end{align*}
where the final equality is the Frobenius formula.  This proves \cref{lem:wishart-schur}.

\section{Normalization checks for the twirl blocks}\label{app:twirl-checks}

The explicit block \eqref{eq:Omega-lambda-explicit} has trace one because
\[
\tr\frac{I_{\mathcal Q_\lambda^A}}{d_\lambda}
=\tr\frac{I_{\mathcal Q_\lambda^B}}{d_\lambda}=1,
\qquad
\norm{\Phi_\lambda^{\mathcal P}}=1.
\]
The total weight in \eqref{eq:twirl-decomp} is
\[
\sum_{\lambda\vdash m}f^\lambda s_\lambda(p)
=\left(\sum_i p_i\right)^m=1
\]
by \eqref{eq:p1-schur}.  Thus \eqref{eq:twirl-decomp} is a convex orthogonal decomposition into density operators.  In the Wishart state averaged over the prior, homogeneity replaces $s_\lambda(p)$ by $R_G^{-m}s_\lambda(W)$, while the $R_G^m$ radial tilt cancels the denominator.  The factor $f^\lambda$ from the probability of the Schur label then cancels the factor $1/f^\lambda$ in \eqref{eq:wishart-schur}, leaving exactly the unnormalized weight $m!s_\lambda(a)^2$ in \cref{prop:schur-bayes}.

\end{document}